\documentclass{article}

\PassOptionsToPackage{numbers, compress,sort}{natbib}

\usepackage[preprint]{nips_2018}

\usepackage[utf8]{inputenc} 
\usepackage[T1]{fontenc}    
\usepackage{hyperref}       
\usepackage{url}            
\usepackage{booktabs}       
\usepackage{amsfonts}       

\usepackage{amssymb}
\usepackage{amsmath,amsthm}
\usepackage{float}
\usepackage{url}
\usepackage{setspace}
\usepackage{hyperref}
\usepackage{todonotes}
\usepackage{amsthm} 
\usepackage[ruled,vlined,linesnumbered]{algorithm2e} 
\usepackage{makecell}
\usepackage{upgreek}
\usepackage{pdfpages}
\usepackage{times}
\usepackage{multirow}
\usepackage[toc,page]{appendix}
\usepackage{listings}
\newtheorem{thm}{Theorem}
\newtheorem{lemma}[thm]{Lemma}
\newtheorem{definition}[thm]{Definition}

\newtheorem{theorem}[thm]{Theorem}

\newtheorem{proposition}[thm]{Proposition}
\newtheorem{corollary}[thm]{Corollary}
\usepackage{color,soul}
\usepackage{tikz}
\usetikzlibrary{arrows,shapes.geometric,shapes.arrows}
\usepackage{pgfplots}
\usepackage{multirow}

\DeclareMathOperator{\support}{support}
\DeclareMathOperator{\Trim}{Trim}

\SetKwFunction{mEqOne}{mEqOne} 
\DeclareMathOperator{\KlmApprox}{KolmogorovApprox}
\DeclareMathOperator{\OptTrim}{OptTrim}

\title{An optimal approximation of discrete random variables with respect to the Kolmogorov distance}

\author{
  Liat Cohen, Gera Weiss \\
  Department of Computer Science\\
  Ben Gurion University of The Negev\\
  \texttt{liati@post.bgu.ac.il, geraw@bgu.ac.il} 
  \And
  Dror Fried \\
  Department of Computer Science \\
  Rice University \\
  \texttt{dror.fried@rice.edu} 
}

\begin{document}

\maketitle

\begin{abstract}
We present an algorithm that takes a discrete random variable $X$ and a number $m$ and computes a random variable whose support (set of possible outcomes) is of size at most $m$ and whose Kolmogorov distance from $X$ is minimal. In addition to a formal theoretical analysis of the correctness and of the computational complexity of the algorithm, we present a detailed empirical evaluation that shows how the proposed approach performs in practice in different applications and domains.
\end{abstract}

\section{Introduction}

Many different approaches to approximation of probability distributions are studied in the literature~\cite{AMCR83,pavlikov2016cvar,PS77,vidyasagar2012metric,cohen2015estimating,CohenGW18}. 
These approaches vary in the types random variables considered, how they are represented, and in the criteria used for evaluation of the quality of the approximations. This paper is on approximating discrete distributions represented as explicit probability mass functions with ones that are simpler to store and to manipulate. This is needed, for example, when a discrete distribution is given as a large data-set, obtained, e.g., by sampling, and we want to represent it approximately with a small table (see~\cite{huq2016efficient,cohen2015estimating} for examples).  

The main contribution of this paper is an efficient algorithm for computing the best possible approximation of a given random variable with a random variable whose complexity is not above a prescribed threshold, where the measures of the quality of the approximation and of its complexity are as specified in the following two paragraphs. 


We measure the quality of an approximation scheme by the distance between random variables and their approximations. Specifically, we use the Kolmogorov distance which is  commonly used for comparing random variables in statistical practice and literature. Given two random variables $X$ and $X'$ whose cumulative distribution functions (cdf) are $F_X$ and $F_{X'}$, respectively, the Kolmogorov distance between $X$ and $X'$ is $d_K(X,X')= \sup_t |F_X(t) - F_{X'}(t)|$ (see, e.g.,~\cite{gibbons2011nonparametric}). We say that $X'$ is a good approximation of $X$ if $d_K(X,X')$ is small. This distance is the basis for the often used Kolmogorov-Smirnoff test for comparing a sample to a distribution or two samples to each other. It represents closeness in terms of deadlines~\cite{cohen2015estimating}.

The complexity of a random variable is measured here by the size of its support, the set of possible outcomes, $m=|\support(X)|=|\{x\colon Pr(X=x) \neq 0\}|$. When distributions are maintained as explicit tables, as done in many implementations of statistical software, the size of the support of the variable is proportional to the amount of memory needed to store it and to the complexity of the computations around it. 

Together, the exact notion of optimality of the approximation targeted in this paper is:
\begin{definition}
	A random variable $X'$ is an optimal $m$-approximation of a random variable $X$ if $|\support(X')| \leq m$ and there is no random variable $X''$ such that $|\support(X'')| \leq m$ and $d_K(X,X'') < d_K(X,X')$.
\end{definition}

In these terms, the main contribution of the paper is an efficient (polynomial time and memory) algorithm that takes $X$ and $m$ as parameters and constructs an optimal $m$-approximation of $X$.

The rest of the paper is organized as follows. In Section~\ref{sec:rel-work} we describe how our work relates to other algorithms and problems studied in the literature. In Section~\ref{sec:alg} we detail the proposed algorithm, analyze its properties, and prove the main theorem. In Section~\ref{sec:exp} we demonstrate how the proposed approach performs on the problem of estimating the probability of hitting deadlines is plans and on randomly generated random variables and compare it to alternatives approximation approaches from the literature. The paper is concluded with a discussion and with ideas for future work in Section~\ref{sec:discussion}.

\section{Related work}
\label{sec:rel-work}
The most relevant work related to this paper is the papers by Cohen at. al.~\cite{cohen2015estimating,CohenGW18}. These papers study approximations of random variables in the context of estimating deadlines. In this context, $X'$ is defined to be a good approximation of $X$ if $F_{X'}(t) > F_{X}(t)$ for any $t$ and $\sup_t F_{X'}(t) - F_{X}(t)$ is small. Note that this measure is not a proper distance measure because it is not symmetric. The motivation given by Cohen at. al. for using this type of approximation is for cases where overestimation of the probability of missing a deadline is acceptable but underestimation is not. In Section~\ref{sec:exp}, we consider the same case-studies examined by Cohen at. al. and show how the algorithm proposed in this paper performs relative to the algorithms proposed there when both over- and under- estimations are allowed. As expected, the Kolmogorov distance between the approximated and the original random variable is considerably smaller when using the algorithm proposed in this paper. 

Another relevant prior work is the theory of Sparse Approximation (aka Sparse Representation) that deals with sparse solutions for systems of linear equations, as follows. 
Given a matrix $D \in \mathbb{R}^{n \times p}$ and a vector $x \in \mathbb{R}^n$, the most studied sparse representation problem is finding the sparsest possible representation $\alpha \in \mathbb{R}^p$ satisfying $x = D\alpha$:
\[
\min_{\alpha \in \mathbb{R}^p} \|\alpha\|_0 \text{ subject to } x = D\alpha
\]
where $\|\alpha\|_0 = |\{ i \in [p]: \alpha_i \neq 0 \}|$ is the $\ell_0$ pseudo-norm, counting the number of non-zero coordinates of $\alpha$. This problem is known to be NP-hard with a reduction to NP-complete subset selection problems.
In these terms, using also the $\ell_\infty$ norm that represents the maximal coordinate and the $\ell_1$ norm that represents the sum of the coordinates, our problem can be phrased as:
\[
\min_{\alpha \in [0,\infty)^p}\|x - D\alpha\|_{\infty} \text{ subject to }  \|\alpha\|_0 = m \text{ and } \|\alpha\|_1=1
\]
where $D$ is the lower unitriangular matrix, $x$ is related to $X$ such that the $i$th coordinate of $x$ is $F_X(x_i)$ where $\support(X)=\{x_1 < \cdots < x_n\}$ and $\alpha$ is related to $X'$ such that the $i$th coordinate of $\alpha$ is $f_{X'}(x_i)$. The functions $F_X$ and $f_{X'}$ represent, respectively, the cumulative distribution function of $X$ and the mass distribution function of $X'$, i.e.,  the coordinates of $x$ are positive and monotonically increasing and its last coordinate is one. We show that this specific sparse representation problem can be solved in $O(n^2m)$ time and $O(m^2)$ memory.

The presented work is also related to the research on binning in statistical inference. Consider, for example, the problem of credit scoring~\cite{zeng2017comparison} that deals with separating good applicants from bad applicants where the Kolmogorov–Smirnov statistic KS is a standard measure. The KS comparison is often preceded by a procedure called binning where small values in the the probability mass function are moved to nearby values. There are many methods for binning~\cite{mays2001handbook,refaat2011credit,bolton2010logistic,siddiqi2012credit}.
In this context, our algorithm can be considered as a binning strategy that provides optimality guarantees with respect to the Kolmogorov distance.

Our study is also related to the work of Pavlikov and Uryasev~\cite{pavlikov2016cvar}, where a procedure for producing a random variable $X'$ that optimally approximates a random variable $X$ is presented. Their approximation scheme, achieved using linear programming, is designed for a different notion of distance called CVaR. The contribution of the present work in this context is that our method is direct, not using linear programming, thus allowing tighter analysis of time and memory complexities. Also, our method is designed for minimizing the Kolmogorov distance that is more prevalent in applications. For comparison, in Section~\ref{sec:exp} we briefly discuss the performance of linear programming approach similar to the one proposed in~\cite{pavlikov2016cvar} for the Kolmogorov distance and compare it our algorithm. 

A problem very similar to ours is termed ``order reduction'' by Vidyasagar in~\cite{vidyasagar2012metric}. There, the author defines an information-theoretic based distance between discrete random variables and studies the problem of finding a variable whose support is of size $m$ and its distance from $X$ is as small as possible (where $X$ and $m$ are given). The only difference between this and the problem studied in this paper, is that Vidyasagar examines a different notion of distance. Vidyasagar proves that computing the distance (that he considers) between two probability distributions, and computing the optimal reduced order approximation, are both NP-hard problems, because they can both be reduced to nonstandard bin-packing problems. He then develops efficient greedy approximation algorithms. In contrast, our study shows that there are efficient solutions to these problems when the Kolmogorov distance is considered.

\section{An algorithm for optimal approximation}\label{sec:alg}
In the scope of this section, let $X$ be a given random variable with a finite support of size $n$, and let  $0<m\leq n$ be a given complexity bound. The section evolves by developing notations and by collecting facts towards an algorithm for finding an optimal $m$-approximation of $X$.

The first useful fact is that it is enough to limit our search to approximations $X'$s such that $\support(X') \subseteq \support(X)$:

\begin{lemma}\label{lem:supContained}
	For every random variable $X''$ there is a random variable $X'$ such that $\support(X') \subseteq \support(X)$ and $d_{K}(X,X')\leq d_{K}(X,X'')$.
\end{lemma}

\begin{proof}
	Let $\{x_1 <\cdots<x_n\} = \support(X)$, and let $x_0 = -\infty, x_{n+1}=\infty$. Consider the random variable $X'$ whose probability mass function is
	$f_{X'}(x_i) = P(x_{i-1} < X'' \leq x_i)$ for $i=1,\dots,n-1$,  $f_{X'}(x_n) = P(x_{n-1} < X'')$, and $F_{X'}(x)=0$ if $x\notin \support(X)$.  Since we only "shifted" the probability mass of $X''$ to the support of $X$, we have that $f_{X'}$ is a probability mass function and therefore $X'$ is well defined. 	
	By construction, $|F_{X}(x_i)-F_{X'}(x_i)| = |F_{X}(x_i)-F_{X''}(x_i)|$ for every  $0 <  i < n$. For $i=n$ we have $|F_{X}(x_n)-F_{X'}(x_n)| = |1-1|=0$.
	Since $|F_{X}(x)-F_{X'}(x)| = |F_{X}(x_i)-F_{X'}(x_i)|$ for every  $0\leq   i  \leq  n$  and $x_i<x<x_{i+1}$, we have that $d_K(X,X')=max_{i}|F_{X}(x_i)-F_{X'}(x_i)|\leq max_{i}|F_{X}(x_i)-F_{X''}(x_i)|\leq d_K(X,X'')$.
\end{proof}


For a set $S\subseteq \support(X)$, let $\mathbb{X}_S$ denote the set of random variables whose supports are contained in $S$. In Step 1 below, we identify a random variable in $\mathbb{X}_S$  that minimizes the Kolmogorov distance from $X$. We denote the Kolmogorov distance between this variable and $X$ by $\varepsilon(X,S)$. Then, in Step 2, we show how to efficiently find a set $S \subseteq \support(X)$ whose size is smaller or equal to $m$ that minimizes $\varepsilon(X,S)$. Then, in Step 3, an optimal $m$-approximation is constructed by taking a minimal approximation in $\mathbb{X}_S$ where $S$ is the set of size $m$ that that minimizes $\varepsilon(X,S)$. 

\subsection*{Step 1: Finding an $X'$ in $\mathbb{X}_S$ that minimizes $d_K(X,X')$}

We first fix a set $S\subseteq \support(X)$ and among all the random variables in $\mathbb{X}_S$ find one with a minimal distance from $X$. Denote the elements of $S$ in increasing order by $S=\{x_1<\dots<x_m\}$ and let $x_0=-\infty$ and $x_{m+1}=\infty$. Consider the following weight function:

\begin{definition}\label{def:weight} For $0\leq i \leq m$ let
	\[
	w(x_i,x_{i+1})=
	\begin{cases}
	P(x_i < X < x_{i+1}) & \text{if $i=0$ or $i = m$;} \\
	P(x_i < X < x_{i+1})/2 & \text{otherwise.} \\	
	\end{cases}
	\]
\end{definition}
For each $1 <  i \leq m$ let $\hat x_i$ be the maximal element of $\support(X)$ that is smaller than $x_i$. 
Note that, for $1\leq i \leq m$, $P(x_i < X < x_{i+1}) = F_X(\hat x_{i+1}) - F_X(x_i)$, a fact that we will use throughout this section.

\begin{definition}\label{def:error}
Let $\varepsilon(X,S) = \max\limits_{i=0,\dots,m} w(x_{i}, x_{i+1})$.
\end{definition}

We first show that $\varepsilon(X,S)$ is a lower bound for the distance between random variable in $\mathbb{X}_S$ and $X$. Then, we present a random variable $X'\in \mathbb{X}_S$ such that $d_K(X,X')=\varepsilon(X,S)$. It then follows that $X'$ is closets to $X$ in $\mathbb{X}_S$.



\begin{proposition}\label{prop:minimal}
	If $X'\in\mathbb{X}_S$ then $d_K(X,X') \geq \varepsilon(X,S)$.
\end{proposition}

\begin{proof}

	By definition, for each $0\leq i\leq m$, $d_K(X,X') \geq \max \{|F_X(\hat x_{i+1}) - F_{X'}(\hat x_{i+1})|,|F_X(x_i) - F_{X'}(x_i)| \}$. Note that $F_{X'}(\hat x_{i+1})=F_{X'}(x_i)$ since the probability of  elements not in $S$ is  $0$.
	
	If $i=0$, that is $x_i=-\infty$, we have that $F_X(x_i)=F_{X'}(x_i)=F_{X'}(\hat x_{i+1})=0$ and therefore $d_K(X,X') \geq |F_X(\hat x_{i+1})| = |F_X(\hat x_{i+1}) - F_{X}(x_i)| =  P(x_i < X < x_{i+1})= w(x_i,x_{i+1})$.
	
	If $i =m$, that is $x_{i+1}=\infty$, we have that $F_X(\hat x_{i+1})=F_{X'}(\hat x_{i+1})=F_{X'}(x_i)=1$. 
	and therefore $d_K(X,X') \geq |1-F_X(\hat x_i)| = |F_X(\hat x_{i+1}) - F_{X}(x_i)| =  P(x_i < X < x_{i+1}) = w(x_i,x_{i+1})$. 
	
	%
	%
	%
	
	Otherwise for every $1\leq i< m$,  we use the fact that $max\{|a|,|b|\} \geq |a-b|/2$ for every $a,b\in\mathbb{R}$, to deduce that $d_K(X,X') \geq 1/2| F_X(\hat x_{i+1}) - F_X(x_i) + F_{X'}(x_i) -F_{X'}(\hat x_{i+1})|$. So $d_K(X,X') \geq 1/2| F_X(\hat x_{i+1}) - F_X(x_i) | =P(x_1 < X < x_2)/2 = w(x_i,x_{i+1})$. 
	
Since $d_K(X,X') {\geq}  w(x_i,x_{i+1})$ for any $0{\leq} i{\leq} m$, the proof follows by the definition of $\varepsilon(X,S)$.
\end{proof}

We next describe a random variable $X'\in\mathbb{X}_S$ in a distance of $\varepsilon(X,S)$ from $X$ by its probability mass function:

\begin{definition}\label{def:construction}
	Let $f_{X'}(x_{i}) = w(x_{i-1},x_i) + w(x_i,x_{i+1}) + f_{X}(x_i)$ for $i=1,\dots,m$ and $f_{X'}(x)=0$ for $x \notin S$.
\end{definition}

We first show that $X'$ is a properly defined random variable:

\begin{lemma}
	$f_{X'}$ is a probability mass function. 
\end{lemma}

\begin{proof}
	By definition $f_{X'}(x_{i})\geq 0$ for every $i$. To see that $\sum_i f_{X'}(x_{i}) =1$, 
	we have $\sum_i f_{X'}(x_{i}) = \sum_i (w(x_{i-1},x_i) + w(x_i,x_{i+1}) + f_{X}(x_i)) = 	
	\sum _{x_i\in S} f_{X}(x_i)) + w(x_0,x_1)+ \sum_{0 < i < m} 2 w(x_i,x_{i+1}) + w(x_m,x_{m+1}) = 
	\sum _{x_i\in S} P(X{=}x_i) + P(x_0 {<} X {<} x_1)+ \sum_{0 < i < m} P(x_i < X < x_{i+1}) +P(x_m < X < x_{m+1}) = 1$ since this is the entire support of $X$.	
\end{proof}

Note that, given $S$, $X'$ can be constructed in time linear in the size of the support of $X$.
Its main property, of course, is that the distance between the cumulative distribution functions of $X$ and of $X'$ are bounded by $w(x_i,x_{i+1})$, as follows:
\begin{lemma}\label{lem:balance}
	Let $x\in\support(X)$ and $0\leq i\leq m$ be such that $x_i\leq x \leq x_{i+1}$ then 
    \[
    -w(x_i,x_{i+1})\leq F_X(x)-F_{X'}(x)\leq  w(x_i,x_{i+1}).
    \]
    \end{lemma}

\begin{proof}	
	By induction on $i$.	First see that $F_{X'}(x) = 0$  for every $x_0<x < x_1$ and therefore $F_X(x)-F_{X'}(x) = F_X(x)-0 \leq F_X(\hat x_1)= F_X(\hat x_1)- F_X(x_0) = w(x_0,x_1)$. For $ x = x_1$ we have $F_X(x_1)-F_{X'}(x_1)=  F_X(\hat x_1) +f_{X}(x_1) - (w(x_0,x_1) + w(x_1,x_2) + f_{X}(x_1) = 
	w(x_0,x_1) +f_{X}(x_1) - (w(x_0,x_1) + w(x_1,x_2) + f_{X}(x_1)) = -w(x_1,x_2)$.
	
	Next assume that $F_X(\hat x_i)-F_{X'}(\hat x_i) = w(x_{i-1},x_i)$.
	Then $F_X(x_i)-F_{X'}(x_i) =  F_X(\hat x_i) +f_{X}(x_i) - (w(x_{i-1},x_i) + w(x_i,x_{i+1}) + f_{X}(x_i)) =  w(x_{i-1},x_i) +f_{X}(x_{i}) - (w(x_{i-1},x_i) + w(x_{i},x_{i+1}) + f_{X}(x_{i})) = -w(x_{i},x_{i+1})$.
	
	As before we have that for all $x_i< x< x_{i+1}$, we have $F_X(x)-F_{X'}(x) = F_X(x)-F_{X'}(\hat x_{i+1}) \leq  F_X(\hat x_{i+1})-F_{X'}(\hat x_{i+1})$. Then $F_X(\hat x_{i+1})-F_{X'}(\hat x_{i+1}) = (F_X(x_i)+ P(x_i < x <x_{i+1})) - F_{X'}(x_i) = -w(x_{i},x_{i+1}) + 2w(x_{i},x_{i+1}) = w(x_{i},x_{i+1}) $.

	Finally, for $x_m\leq x \leq x_{m+1}$ we have that $F_{X'}(x_m)=1$ therefore $F_X(x_m)-F_{X'}(x_m) = (1- P(x_m<X < x_{m+1})) - 1 =  P(x_m<X < x_{m+1}) = w(x_m,x_{m+1})$, and for every $x_m<x<x_{m+1}$ we have $F_X(x) - F_{X'}(x) < (1- P(x_m<X < x_{m+1})) - 1 <  - P(x_m<X < x_{m+1}))  = -w(x_m,x_{m+1})$ as required.
\end{proof}

From Lemma \ref{lem:balance}, by the definition of $\varepsilon(X,S)$, we then have:
\begin{corollary} \label{col:Xprime}
	$d_K(X,X') = \varepsilon(X,S)$.
\end{corollary}
From Proposition~\ref{prop:minimal} we also have:
\begin{corollary} \label{col:Xprimeopt}
	$\varepsilon(X,S)$ is the Kolmogorov distance between $X$ and the variables closest to it in $\mathbb{X}_S$.
\end{corollary}

\subsection*{Step 2: Finding a set $S$ that minimizes $\varepsilon(X,S)$}

We proceed to finding a set $S$ of size $m$ that minimizes $\varepsilon(X,S)$. To obtain that we use a graph search approach inspired by~\cite{chakravarty1982partitioning}. We construct a directed graph with a source and a target in which each source-to-target path of length smaller or equal to $m$ corresponds to a possible support set of the same size, and the weights along that path correspond to the weight as defined in Definition~\ref{def:weight}. Thus the problem of finding an $S$ of size $m$ that minimizes $\varepsilon(X,S)$ is reduced to the problem of finding a source-to-target path $\vec{p}$ of length smaller or equal to $m$ in that graph such that the maximal weight of an edge in $\vec{p}$ is minimal among all other such maximal edges in all other such paths.

More specifically, the vertices of the graph are $V=\support(X) \cup \{-\infty,\infty\}$ and the edges, $E$, are all the pairs $(x_1,x_2) \in V^2$ such that $x_1 < x_2$. The weight 
of each edge is as specified in Definition~\ref{def:weight}. Note that there is a one-to-one correspondence between a set $S \subseteq \support(X)$ of size $m$, and an $-\infty$-to-$\infty$ path $\vec{p}_S$ in $G$, obtained by removing the $-\infty$ and $\infty$ from the path in one direction and by adding these elements and sorting in the other direction. 
With this correspondence, the maximal weight of an edge on $\vec{p}_S$ is $\varepsilon(X,S)$. We denote this maximal weight of an edge  by $w(\vec{p}_S)$, and denote the set of all acyclic $-\infty$-to-$\infty$ paths in $G$ with at most $m$ edges by $paths_m(G, -\infty, \infty)$. Thus, the problem of finding the set $S$ with the minimal  $\varepsilon(X,S)$ is now reduced to the problem of finding a path $\vec{p}\in paths_m(G, -\infty, \infty)$ such that $w(\vec{p})$ is minimal among all $\{w(\vec{p'}) \colon \vec{p'}\in paths_m(G, -\infty, \infty)\}$. 
This problem can be solved by a variant of the Bellman-Ford algorithm and by the algorithm described in~\cite{guerin2002computing} that performs better in some cases.


\subsection*{Step 3: Constructing the overall algorithm}

We combine Step 1 and Step 2 in the following algorithm called $\KlmApprox$ (Algorithm~\ref{alg:optapprox}) that follows naturally from the two steps. Given $X$ and $\support(X)$ we add $x_0=\infty,x_{n+1}=-\infty$ and construct the graph (line 2) as described  in Step 2 above. Then we execute a variant of the Bellman-Ford algorithm on $G$ for $m$ iterations, or the algorithm proposed in~\cite{guerin2002computing}, to obtain a path $\vec{p}$ (line 2). Finally, we use Definition~\ref{def:construction} to construct $X'$ from $\vec{p}$ (line 3).

\begin{algorithm}\label{alg:optapprox}
	\DontPrintSemicolon
	Construct a weighted graph $G=(V,E)$ where $V=\support(X) \cup \{-\infty,\infty\}$, $E=\{(x_1,x_2) \in V^2\colon x_1 < x_2 \}$, and the weights are as in Definition~\ref{def:weight}.\; \vspace{0.2cm}

	Find a path $\vec{p}=(x_0,\dots,x_{m+1})\in paths_m(G, -\infty, \infty)\}$ such that $w(\vec{p}) = \min \{w(\vec{p}) \colon \vec{p}\in paths_m(G, -\infty, \infty)\}$.   \; 	\vspace{0.2cm}

	Return a random variable whose probability mass function is 
		$f_{X'}(x_{i}) = w(x_{i-1},x_i) + w(x_i,x_{i+1}) + f_{X}(x_i)$ for all $i=1,\dots,m$ and zero otherwise. \;
	
	\caption{$\KlmApprox (X, m)$}  
	\label{alg:sequence}
\end{algorithm}

%

\begin{theorem}\label{the:algo}
	$\KlmApprox$ returns an $m$-optimal-approximation of $X$.
\end{theorem}
\begin{proof}
By the construction of $G$ we get that the path $\vec{p}$ obtained in line 2 of  $\KlmApprox$ describes a set $S$ of support of size at most $m$ for which $\varepsilon(S,X)$ is minimal. Then from Definition \ref{def:construction} and Corollary \ref{col:Xprime} we construct $X'$ in lines 3 of 	$\KlmApprox$ such that $d_K(X,X') = \varepsilon(X,S)$. Therefore $X'$ is closets to $X$ among all random variables with support contained in $\support(X)$. From Lemma \ref{lem:supContained} we then get that $X'$ is an $m$-optimal-approximation of $X$.
\end{proof}

The memory and time  complexities of $\KlmApprox$ are as follows.

\begin{theorem}\label{the:complexity}
	The $\KlmApprox(X,m)$ algorithm runs in time $O(n^2m)$, using $O(n^2)$ memory where $n=|\support(X)|$.
\end{theorem}
\begin{proof}
	Constructing the graph $G$ as described in Step 2 takes $O(n^2)$ time and memory. Computing the shortest path can be achieved, for example, by the algorithm described in~\cite{guerin2002computing} in time $O(n^2 m)$ and no additional memory allocation.
\end{proof}

\section{Experimental evaluation}\label{sec:exp}

We describe below several experiments that show how $\KlmApprox$ performs in practice in different applications and domains.
All algorithms were implemented in Python and the experiments were executed on a hardware comprised of an Intel i5-6500 CPU @ 3.20GHz processor and 8GB memory. The algorithms of Cohen at. al. were taken "as is" from in the supplementary material to~\cite{cohen2015estimating}.

\paragraph{Repetitive support size minimization.} One use of support size minimization is when commutations that involve summations of random variables slow due to an exponential growth in the support of convolutions of random variables~\cite{cohen2015estimating}. A key action in coping with this situation is reduction of the  support size by replacing the summed random variable by an approximation of it that has a smaller support size. Previous work such as the work of Cohen at. al. in ~\cite{cohen2015estimating,CohenGW18} handle this reduction using weaker or sub-optimal notion of approximation than ours, as discussed in Section~\ref{sec:rel-work}. 

As proven in Section~\ref{sec:alg}, given $m$, a single step of $\KlmApprox$ guarantees an optimal $m$-approximation. However in the setting considered here we need to repetitively use $\KlmApprox$, thus the optimality of the eventually obtained random variable is not guaranteed. In light of this, we tested the accuracy of the repetitive-$\KlmApprox$ to see how it performs against the tools of~\cite{cohen2015estimating,CohenGW18} using their benchmarks. These benchmarks are taken from the area of task trees with deadlines, a sub area of the well-established Hierarchical planning~\cite{thomas1988hierarchical, alford2016hierarchical, xiao2017hierarchical}.

We estimated the probability for meeting deadlines in plans, as described in~~\cite{cohen2015estimating,CohenGW18}, and experimented with four different methods of approximation. The first two, $\OptTrim$~\cite{CohenGW18} and the $\Trim$~\cite{cohen2015estimating}, are taken from the repository provided by the authors and are designed for achieving only a one-sided Kolmogorov approximation - a  weaker notion of approximation then the Kolmogorov approximation analyzed in this work. The third method is a simple sampling scheme also described in~\cite{cohen2015estimating} and the fourth is our Kolmogorov approximation obtained by the proposed $\KlmApprox$ algorithm. The parameters for the different methods were chosen in a compatible way, as explained in~\cite{CohenGW18}. We ran also an exact computation as a reference to the approximated one in order to calculate the errors. 

\begin{table}[th]
	\scriptsize
	\centering
	\renewcommand{\arraystretch}{1.3}
	\begin{tabular}{|c|c|c|c|c|c|c|}
		\hline
		\multirow{2}{*}{Task Tree} & \multirow{2}{*}{$M$} & {$\KlmApprox$} & {$\OptTrim$} & {$\Trim$} & \multicolumn{2}{c|}{Sampling} \\ \cline{3-7} 
		&	& $m/N{=}10$ & $m/N{=}10$ & $\varepsilon\cdot N{=}0.1$ & $s{=}10^{4}$& $s{=}10^{6}$ \\ \hline
		\hline
		
		
		\multirow{2}{*}{Logistics} & 2& 0 & 0 &  0.0019 &  0.007 & 0.0009  \\ \Xcline{2-7}{1pt}
		{\tiny $(N=34)$}& 4& 0.0024 & 0.0046&  0.0068  &   0.0057 & 0.0005 \\\Xhline{1pt}
		
		\multirow{2}{*}{DRC-Drive}  
		&2	& 0.0014 & 0.004&  0.009  & 0.0072 & 0.0009  
		\\ \Xcline{2-7}{1pt}
		
		{\tiny $(N{=}47)$}& {4}& 0.001 & 0.008&  0.019   & 0.0075  & 0.0011 
		\\  \Xhline{1pt}

		\multirow{2}{*}{Sequential}  & {2} & 0.0093 & 0.015 &  0.024 & 0.0063 & 0.0008 \\ \Xcline{2-7}{1pt}  
		{\tiny $(N{=}10)$} & {4} & 0.008 & 0.024 &  0.04 & 0.008 & 0.0016 \\ \Xhline{1pt}

	\end{tabular}
	\caption{Comparison of estimated errors with respect to the reference exact computation on various task trees.}
	\label{tab:errors}
\end{table}

Table~\ref{tab:errors} shows the results of the experiment. The quality of the solutions obtained with the $\KlmApprox$ operator are better than those obtained by the $\Trim$ and $\OptTrim$ operators as expected. In some of the task trees, the sampling method produced better results than the approximation algorithm with $\KlmApprox$. Still, the $\KlmApprox$ approximation algorithm comes with an inherent advantage of providing an exact quality guarantees, as opposed to sampling where the best one can hope for is probabilistic guarantees.

\paragraph{Single step support minimization.}
In order to better understand the quality gaps in practice between $\KlmApprox$, $\OptTrim$, and $\Trim$, we tested their relative errors when applied on single random variables with support size $n = 100$, and different $m$s. Note that the error obtained by $\KlmApprox$ is optimal while the other methods are not optimized for the Kolmogorv distance. In each instance of this experiment, a random variable is randomly generated by choosing the probabilities of each element in the support uniformly and then normalizing these probabilities so that they sum to one.

Figure~\ref{fig:error} presents the error produced by the above methods. The depicted results are averages over fifty instances of random variables. The curves in the figure show the average error of $\OptTrim$ and $\Trim$ operators with comparison to the average error of the optimal approximation provided by $\KlmApprox$ as a function of $m$. It is evident from this graphs that increasing the support size of the approximation $m$ reduces the error, as expected, in all three methods. However, the (optimal) errors produced by the $\KlmApprox$ are significantly smaller, a half of the error produced by $\OptTrim$ and $\Trim$.

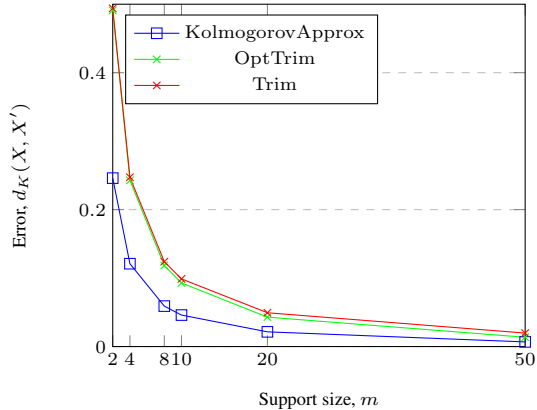
\begin{figure}[htb]
	\scriptsize	
	\centering 
	\begin{tikzpicture}
	\begin{axis}[
	scale=.8,
	xlabel={Support size, $m$},
	ylabel={Error, $d_K(X,X')$},
	xmin=2, xmax=50,
	ymin=0, ymax=0.5,
	xtick={2,4,8,10,20,50},
	legend pos=north west,
	ymajorgrids=true,
	grid style=dashed,
	]
	
	\addplot[
	color=blue,
	mark=square,
	]
	coordinates { 
		(2 , 0.246) 
		(4 , 0.121) 
		(8 , 0.0591) 
		(10 , 0.046) 
		(20, 0.0215) 
		(50, 0.0068) 
		
	};
	\addlegendentry{$\KlmApprox$}
	
	\addplot[
	color=green,
	mark=x,
	]
	coordinates {
		(2 , 0.491) 
		(4 , 0.2428) 
		(8 , 0.1184) 
		(10 , 0.0929) 
		(20, 0.0430) 
		(50, 0.0136) 
	};
	\addlegendentry{$\OptTrim$}
	
	\addplot[
	color=red,
	mark=x,
	]
	coordinates {
		(2 , 0.494) 
		(4 , 0.2473) 
		(8 , 0.124) 
		(10 , 0.0988) 
		(20, 0.0494) 
		(50, 0.01971)  
	};
	\addlegendentry{$\Trim$}
	
	\end{axis}
	\end{tikzpicture}
	\caption{Error comparison between $\KlmApprox$, $\OptTrim$, and $\Trim$ on randomly generated random variables as function of $m$.}
	\label{fig:error}
\end{figure}

\paragraph{Comparison to Linear Programming.}
We also compared the run-time of $\KlmApprox$ with a linear programing (LP) algorithm that also guarantees optimality, as described and discussed for example in~\cite{pavlikov2016cvar}.
For that, we used the ``Minimize'' function of Wolfram Mathematica as a  state-of-the-art implementation of linear programing, encoding the problem by the LP problem $\min_{\alpha \in \mathbb{R}^n} \| x - \alpha\|_\infty$ subject to $\|\alpha\|_0 \leq m$ and $\| \alpha \|_1 =1$.
The run-time comparison results were clear and persuasive: $\KlmApprox$ significantly outperforms the LP algorithm. For a random variable with support size $n=10$ and $m=5$, the LP algorithm run-time was $850$ seconds, where the $\KlmApprox$ algorithm run-time was less than a tenth of a second. For $n=100$ and $m=5$, the $\KlmApprox$ algorithm run-time was 0.14 seconds and the LP algorithm took more than a day. 
Since it is not trivial to formally analyze the run-time of the LP algorithm, we conclude by the reported experiment that in this case the LP algorithm might not be as efficient as $\KlmApprox$ algorithm whose complexity is proven to be polynomial in Theorem~\ref{the:complexity}.

\section{Discussion and future work}\label{sec:discussion}

We developed an algorithm for computing optimal approximations of random variables where the approximation quality is measured by the Kolmogorov distance.
As demonstrated in the experiments, our algorithm improves on the approach of Cohen at. al.~\cite{cohen2015estimating} and~\cite{CohenGW18} in that it finds an optimal two sided Kolmogorov approximation, and not just one sided. Beyond the Kolmogorov measure studied here we believe that similar approaches may apply also to total variation, to the Wasserstein distance, and to other measures of approximations. Another direction for future work is extensions to tables that represent other objects, not necessarily random variables. To this end, we need to extend the algorithm to support tables that do not always sum to one and tables that may contain negative entries.

\bibliography{library,Trim_Optimum}{}
\bibliographystyle{abbrv}

\end{document}